\documentclass[hidelinks]{ifcolog}
\usepackage[utf8]{inputenc}
\usepackage{amsmath}
\usepackage{amsfonts}
\usepackage{amssymb}
\usepackage{graphicx}
\usepackage{amsthm}
\usepackage[ruled,lined,linesnumbered,noend]{algorithm2e} 
\usepackage{xcolor}
\usepackage{authblk}
\usepackage{subfigure}
\usepackage{listings}

\definecolor{mygreen}{rgb}{0,0.6,0}
\definecolor{mygray}{rgb}{0.85,0.85,0.85}
\definecolor{mymauve}{rgb}{0.58,0,0.82}
\newtheorem{theorem}{Theorem}
\newtheorem{lemma}[theorem]{Lemma}

\newtheorem{definition}{Definition}

\SetKwInput{KwInput}{Input}                
\SetKwInput{KwOutput}{Output}              
\title{An Analysis of the Correctness and Computational Complexity of Path Planning in Payment Channel Networks}
\addauthor[corcoranp@cardiff.ac.uk]
  {Padraig Corcoran}
  {School of Computer Science \& Informatics, Cardiff University, Cardiff, Wales, UK.}
\addauthor[lewisr9@cardiff.ac.uk]
  {Rhyd Lewis}
  {School of Mathematics, Cardiff University, Cardiff, Wales, UK.}

\authorrunning{Corcoran \& Lewis}
\titlerunning{An Analysis of the Correctness and Computational Complexity}
\titlethanks{}

\begin{document}
\maketitle
\def\ed#1{{\color{blue}{\bf [Ed:} {\it{#1}}{\bf ]}}}
\def\Ed#1{{\color{blue}{\bf [Ed:} {\it{#1}}{\bf ]}}}

\begin{abstract}
Payment Channel Networks (PCNs) are a method for improving the scaling and latency of cryptocurrency transactions. For a payment to be made between two peers in a PCN, a feasible low-fee path in the network must be planned. Many PCN path planning algorithms use a search algorithm that is a variant of Dijkstra's algorithm. In this article, we prove the correctness and computational complexity of this algorithm. Specifically, we show that, if the PCN satisfies a consistency property relating to the fees charged by payment channels, the algorithm is correct and has polynomial computational complexity. However, in the general case, the algorithm is not correct and the path planning problem is NP-hard. These newly developed results can be used to inform the development of new or existing PCNs amenable to path planning. For example, we show that the Lightning Network, which is the most widely used PCN and is built on the Bitcoin cryptocurrency, currently satisfies the above consistency property. As a second contribution, we demonstrate that a small modification to the above path planning algorithm which, although having the same asymptotic computational complexity, empirically shows better performance. This modification involves the use of a bidirectional search and is empirically evaluated by simulating transactions on the Lightning Network.
\end{abstract}

\smallskip
\noindent \textbf{Keywords.} payment channel networks; path planning; lightning network.

\section{Introduction}
A cryptocurrency is a digital currency that is decentralised in nature and not issued or controlled by a central authority. The most famous and successful example of a cryptocurrency is Bitcoin, invented by Satoshi Nakamoto~\cite{nakamoto2008bitcoin}. Cryptocurrencies are enabled by the use of blockchain systems, which are a type of distributed system that maintains a ledger or database of digital currency transactions. A necessary component of any blockchain system is a \emph{consensus algorithm}. This is an algorithm that allows all peers in the system to achieve agreement regarding the state of the ledger \cite{xiao2020survey}. As a consequence of the need to ensure decentralisation and security, many consensus algorithms exhibit scalability issues relating to low transaction throughput (the number of transactions per second) and transaction confirmation latency (the time until a transaction can be considered settled). The former property can, in turn, lead to high transaction fees when large sets of transactions are competing to get into the blockchain. For example, Bitcoin can currently only process seven transactions per second and has a transaction confirmation latency of approximately ten minutes. This is assuming that the transaction is actually included in the very next block and is considered confirmed immediately when this happens. In contrast, the VISA payment system can process 65,000 transactions per second and has a transaction confirmation latency of just a few seconds. The low transaction throughput of Bitcoin means that, in times of high demand, transaction fees can be expensive and greatly exceed the cost of many everyday purchases.

To overcome the above scaling challenge, several potential solutions have been proposed \cite{hafid2020scaling}. One promising direction of research is \textit{payment channel networks} (PCNs) that allow the consolidation of a larger set of transactions into a smaller set, where individual transactions can be confirmed almost instantly. The Lightning Network (LN) is the most widely used PCN and is designed to work with Bitcoin~\cite{antonopoulos2021mastering}. To illustrate the potential benefits of PCNs, consider the case where Alice is a merchant and Bob is a customer. Alice makes regular small payments to Bob, and sometimes Bob makes refunds to Alice. Instead of submitting each transaction to be processed immediately by the consensus algorithm, Alice and Bob can establish a payment channel between them that maintains a balance that is updated after each transaction. When Alice and Bob have completed their transactions, they submit the final balance to be processed by the consensus algorithm. Since only a single transaction is processed by the consensus algorithm, this can increase transaction throughput significantly. Furthermore, since each transaction is only processed locally, the transaction latency is reduced to the time required to send a small number of network messages between Alice and Bob. By establishing a network of payment channels, individuals can also make payments to others who they do not share a direct payment channel with. For example, consider the case where Alice wants to make a payment to Charlie but has not established a corresponding payment channel. If Alice has an existing payment channel to Bob, and Bob has an existing payment channel to Charlie, then Alice can make the payment to Bob and Bob can, in turn, forward the payment to Charlie. PCNs are typically implemented using smart contracts, which are pieces of code executed on a blockchain that ensure the security of the payments.

Despite the potential of PCNs, the technology still faces several challenges. This work focuses on path planning which concerns finding a useful path in the network of payment channels for making a given payment. There exists no universal method for performing path planning in most PCNs including the LN. Instead, different implementations use different methods for solving this problem \cite{decker2018eltoo}. Given that the size of PCNs can be very large, it is important that path planning is performed in a correct and computationally efficient manner. Most current path planning algorithms use a variant of Dijkstra's algorithm. Dijkstra's algorithm assumes that the network in question contains arcs with constant and non-negative weights. This assumption does not hold in the LN where arc weights correspond to fees charged for forwarding payments and vary as a function of the payment amount. Hence, existing correctness and computational complexity results for Dijkstra's algorithm do not generalise to this variant. In fact, to the authors' knowledge, there exists no analysis of this algorithm's correctness and computational complexity in this setting. In this work, we address this research gap. These new results tell us that, if the PCN in question satisfies a consistency property relating to the fee charged for forwarding payments, the algorithm is correct and has polynomial computational complexity. Otherwise, in the general case, the algorithm is not correct and the path planning problem is NP-hard. In turn, we show that the LN satisfies this property. A second contribution made in this article is the proposal of a small modification to the above path planning algorithm which, though having the same asymptotic computational complexity, empirically shows better performance. This modification involves the use of a bidirectional search.

The remainder of this paper is structured as follows. In Section~\ref{sec:related_works} we review related works on path planning in PCNs. In Section~\ref{sec:problem_def} we formally define the problem of path planning in the LN. This section also identifies a connection between LN path planning and a specific \emph{recurrence relation}, which we solve. In Sections~\ref{sec:uni_path_find} and \ref{sec:formal_analysis} we formally describe the variant of Dijkstra's algorithm mentioned above and provide an analysis of its correctness and computational complexity. In Section~\ref{sec:bi_path_find} we then describe the proposed modification to this algorithm that uses bidirectional search and prove its correctness and computational complexity. In Section~\ref{sec:results} we then demonstrate the computational benefits of this method empirically. Finally, in Section~\ref{sec:conclusions} we conclude this work.

\section{Related Works}
\label{sec:related_works}
As discussed in the introduction to this article, a variant of Dijkstra's algorithm is a path planning algorithm frequently used with PCNs. Given the importance of path planning, several other algorithms have also been proposed. The remainder of this section presents a review of these algorithms followed by a survey of other relevant related topics. A more in-depth review can be found in \cite{rebello2024survey}.

Solutions to the path planning problem can broadly be categorised as centralised methods and decentralised methods. In centralised methods, a single peer independently computes the entire payment path. An application of centralised methods, commonly known as source-based routing, involves the payment sender computing a payment path and using onion routing to make the payment using this path. Onion routing ensures that each peer in the path only learns the identity of the peer immediately before and the peer immediately after them in the path. They do not learn the payment sender, receiver or amount. However, since the path planning task is performed exclusively by the payment sender, source-based routing requires that the sender has significant computational resources and knowledge of the current state of the network. A variant of Dijkstra's algorithm is the most commonly used centralised path planning method. In fact, this method is used by most LN node implementations including LND and Core Lightning which are the most commonly used implementations. Several authors have also proposed centralised path planning methods that compute a set of distinct paths between the source and destination vertices in PCNs \cite{sivaraman2020high, bagaria2020boomerang, pickhardt2021optimally}. This allows multi-part payments to be made, whereby a single payment is split into a set of smaller payments, and each is routed using a distinct path. The main motivations for multi-part payments are to allow larger payments to be made given the finite capacity of individual payment channels and to obfuscate the total amount being transferred.

In decentralised path planning methods, each peer that the payment passes through learns the payment receiver and uses this information to independently decide the next peer to forward the payment to until it is eventually arrives at the receiver. Decentralised methods do not require the payment sender to have significant computational resources or knowledge of the current state of the network. However, since each peer in the payment path learns the identity of the payment receiver, they provide relatively poor privacy. Hence decentralised methods are not used by most LN node implementations. Roos et al.~\cite{roos2017settling} proposed a greedy decentralised path planning algorithm entitled SpeedyMurmurs. This algorithm assigns coordinates to each vertex in the network such that a path from a given source vertex can be obtained by iteratively selecting arcs to vertices closer to the destination vertex. Similarly, Prihodko et al.~\cite{prihodko2016flare} proposed a decentralised path planning algorithm called Flare. In this method, each peer maintains a routing table that contains payment paths to both peers in their local neighbourhood and a small subset of other peers known as beacons. To determine a payment path, the routing tables corresponding to both the sender and receiver are considered. If both routing tables contain a common peer, a payment path is constructed by combining payment paths to and from this peer. Lin et al.~\cite{lin2020rapido} proposed a path planning method that partitions the PCN in a set of regions and elects a peer in each region to be a corresponding beacon. Using this method, path planning is performed by forwarding the payment in question to a beacon. This beacon then forwards the payment to the receiver if the receiver is contained in their respective region. Otherwise, the beacon forwards the payment to the beacon corresponding to the region within which the receiver is contained.

To transfer a given amount along a given path in a PCN, all the channels in question must have sufficient balance or liquidity in the direction in question. To help overcome this challenge, several rebalancing methods have been proposed to adjust channel balances to make paths feasible. Broadly speaking, these methods involve peers making cyclic payments to themselves that move a portion of the balance on one of their channels to another of their channels \cite{pickhardt2020imbalance}. For privacy reasons, the channel balances are not shared among peers in a PCN because by monitoring changes in channel balances, an adversary can infer the source, destination and amount of all payments. A consequence of this is that path planning reduces to a trial and error process whereby different paths are attempted until a feasible one with sufficient balance to perform the payment in question is found. Tang et al.~\cite{tang2020privacy} considered whether it is possible for channels to share noisy balances while still maintaining privacy and found this to be challenging. Dotan et al.~\cite{dotan2022twilight} proposed a method that is robust to balance probing attacks that attempt to determine arc balances. This is achieved by adding noise to a channel's response to a request to forward a payment.

In a PCN, when a channel forwards a payment intended for a third party, they charge a fee for providing this service. Several methods have been proposed to automatically determine a suitable fee to charge, considering objectives such as maximising profit or maintaining a desired channel balance \cite{di2018routing, ersoy2020profit, van2021merchant}. However, by monitoring changes in fees charged and combining this with knowledge of how fees are determined, an adversary can potentially infer information about the transactions being processed by different channels. Tochner et al.~\cite{tochner2019hijacking} demonstrated that, by exploiting the fact that most path planning algorithms select paths with the lowest fees, an adversary can establish channels with low fees that a significant percentage of payments use. These channels can then be used to perform several attacks, including a denial-of-service (DoS) attack.

Several works have also analysed the topological structure of PCNs. Martinazzi et al.~\cite{martinazzi2020evolving} and Seres et al.~\cite{seres2020topological} demonstrated that the LN tends to exhibit small-world network properties with a small proportion of nodes with high centrality that act as hubs. Rohrer et al.~\cite{rohrer2019discharged} examined the resilience of the LN with respect to random failures and targeted attacks. The authors also demonstrated that the LN exhibits small-world and scale-free network properties. Zabka et al.~\cite{zabka2024centrality} performed a centrality analysis of the LN. Similarly, Kotzer et al.~\cite{kotzerbraess2023} demonstrated that Braess's paradox may exist in a PCN, whereby adding additional channels may negatively impact the ability to successfully make payments.

\section{Problem Definition}
\label{sec:problem_def}
In this work, we model a snapshot of the Lightning Network (LN) at a given moment in time as a directed graph $G= (V, E)$. We use the term ``snapshot'' because both the topology and channel properties of the LN are dynamic. The topology is dynamic in the sense that, over time, new channels can be created and existing ones can be deleted. LN channel properties are also dynamic in the sense that, over time, the properties described below of a given channel can change.

In our model, each LN peer is modelled as a vertex $v \in V$. Since LN channels are bidirectional, each LN channel between the vertices $v_i$ and $v_j$ is modelled as a pair of arcs $(v_i, v_j) \in E$ and $(v_j, v_i) \in E$. Each arc has a balance that equals the maximum payment amount that can be transferred along that arc and is modelled as a map $b: E \rightarrow \mathbb{Z}^{\geq}$. For a given LN channel, the sum of the corresponding pair of arc balances is known as the capacity of the channel in question. Although the capacity of a given channel remains constant, the corresponding arc balances change as amounts are transferred between the vertices in question. As an example, consider the pair of arcs $(v_i, v_j) \in E$ and $(v_j, v_i) \in E$ corresponding to a channel between the vertices $v_i$ and $v_j$. If the amount $a \in \mathbb{Z}^{>}$ is transferred from $v_i$ to $v_j$ using this channel, the balance $b(v_i, v_j)$ reduces by $a$ while the balance $b(v_j, v_i)$ increases by $a$. To transfer the amount $a$ along the arc $(v_i, v_j)$, the condition $a \leq b(v_i, v_j)$ must therefore be satisfied. That is, the arc in question must have a sufficient balance.

If one wishes to transfer an amount $a$ along an arc $(v_i, v_j) \in E$, the arc will also charge a fee for this service. This fee is charged by the arc operator, which is the source vertex $v_i$, and is parameterised by two values: a \emph{base fee} and a \emph{fee rate}. The base fee is a fixed value, while the fee rate is a proportion of the amount being transferred. We model the base fees and fee rates as maps $f_b: E \rightarrow \mathbb{Z}^{\geq}$ and $f_r: E \rightarrow \mathbb {R}^{\geq}$ respectively. In turn, we model the total fee charged by an arc $e \in E$ to transfer an amount $a \in \mathbb{Z}$ as the map $f_{br}: E \times \mathbb{Z}^{\geq} \rightarrow \mathbb{Z}^{\geq}$:
\begin{equation} \label{eq:f_br}
	f_{br}(e,a) = f_b(e) + f_r(e) \times a.
\end{equation}
To illustrate the above features of the LN, consider the toy LN displayed in Figure~\ref{fig:ln_eg_1}. This graph contains a single arc $(v_i, v_j)$ that has a base fee of 2, a fee rate of 0.1 and a balance of 20. Consider the case where we wish to transfer an amount 10 from $v_i$ to $v_j$ along this arc. This amount is less than or equal to the arc balance making the transfer feasible. To perform this transfer, vertex $v_i$ charges a fee of $2 + 0.1 \times 10 = 3$. An important point to highlight is that, to transfer an amount $a$ from a vertex $v_i$ to a vertex $v_j$ using an arc $(v_i, v_j)$, we must first transfer $a$ plus the fee in question to $v_i$. In this case, a value of $10+3=13$ is therefore transferred to $v_i$.

\begin{figure}
	\begin{center}
		\includegraphics[height=1cm]{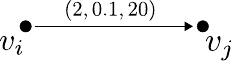}
		\caption{An example LN containing two vertices and a single arc is displayed. The arc is labelled with a tuple containing the corresponding base fee, fee rate, and balance respectively.}
		\label{fig:ln_eg_1}
	\end{center}
\end{figure}

At any given time, the LN will typically be a sparse graph with only a small proportion of vertex pairs being connected by an arc. If we wish to transfer an amount from a source vertex $s$ to a destination vertex $t$ it may be the case that no arc exists from $s$ to $t$, or that this arc has an insufficient balance. On the other hand, it may be possible to transfer the amount along a sequence of arcs $(v_1, v_2), (v_2, v_3), \dots, (v_{n-1}, v_n)$ where $s=v_1$ and $t=v_n$. Such a sequence is known as a \emph{path} from $s$ to $t$. For example, consider the toy LN snapshot displayed in Figure~\ref{fig:ln_eg_2} and the case where we wish to transfer a given amount from $v_i$ to $v_j$. Here, there exists no arc from $v_i$ to $v_j$ but the path $(v_i, v_k), (v_k, v_j)$ may potentially be used to perform the transfer in question.

\begin{figure}
\begin{center}
\includegraphics[height=1cm]{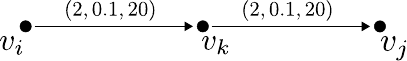}
\caption{An example LN containing three vertices and two arcs is displayed. Each arc is labelled with a tuple containing the corresponding base fee, fee rate and balance.}
\label{fig:ln_eg_2}
\end{center}
\end{figure}

To transfer an amount $a$ from vertex $s$ to vertex $t$ along a path, one must transfer to each vertex in that path the amount $a$ plus the sum of fees for subsequent arcs in the path. Consider again the LN displayed in Figure~\ref{fig:ln_eg_2} and the case where we wish to transfer an amount $a$ from $v_i$ to $v_j$ along the path $e_1, e_2$ where $e_1 =(v_i, v_k)$ and $e_2=(v_k, v_j)$. To successfully transfer the amount $a$ to $v_j$, we must transfer the amount $a + f_{br}(e_2, a)$ to $v_k$ and the amount $a + f_{br}(e_2, a) + f_{br}(e_1, a + f_{br}(e_2, a))$ to $v_i$. Furthermore, the balance of the arc $e_2$ must be greater than or equal to $a$, and the balance of the arc $e_1$ must be greater than or equal to  $a + f_{br}(e_2, a)$.

Let $e_1, e_2, \dots, e_n$ be a path from $s$ to $t$ and let $v_1, v_2, \dots, v_{n+1}$ be the sequence of $n+1$ vertices in this path. Also, let $a_i$ be the amount that must be transferred to vertex $v_i$ such that, ultimately, the amount $a$ is successfully transferred to $t$. The sequence of values $a_{n+1}, a_{n}, \dots, a_{1}$ is defined by the recurrence relation in Equation~(\ref{eq:recurrence_relation}) with the condition $a_{n+1}=a$ \cite[Chapter~7]{charalambides2002enumerative}. In the appendix of this article, we solve this recurrence relation to give a closed-form expression for each $a_i$ value.
\begin{equation} \label{eq:recurrence_relation}
	a_{i-1} = a_i + f_{br}(e_i, a_i).
\end{equation}

We refer to the problem of finding a path from $s$ to $t$ that minimises the total arc fees while satisfying the necessary balance constraints as the \emph{path planning problem}. Most solutions to the path planning problem use \emph{source-based path planning} whereby the payment source performs the path planning task. Note that, many implementations of path planning in the LN consider factors relating to both arc fees and arc reputation \cite{kumble2021lightning, andreescu2021comparing}. However, in this work, we do not consider factors relating to arc reputation because these cannot be determined without access to historical transaction data.

\section{Variant of Dijkstra’s algorithm}
\label{sec:uni_path_find}
Dijkstra's algorithm is a well-known path planning algorithm that constructs a lowest cost path from a source vertex $s$ to a destination vertex $t$ \cite{sanders2019sequential}. It achieves this by maintaining a map $c: V \rightarrow \mathbb{R}$ that stores an upper bound on the cost from $s$ to the vertex in question. The algorithm visits vertices in the order of their cost from the vertex $s$ such that when a vertex is visited the above bound becomes tight. Dijkstra's algorithm terminates when the vertex $t$ is visited. When this happens, we obtain a path of lowest cost from $s$ to $t$ by backtracking through the exploration steps.

In its basic form, Dijkstra's algorithm cannot be directly applied to solve the LN path planning problem. This is because, to make a payment of amount $a$ from $s$ to $t$, it is first necessary to know the amount $a' \geq a$ that must be transferred initially to $s$ such that the amount $a$ is transferred to $t$ after all fees have been subtracted. Let $a_v$ be the amount transferred to vertex $v$ so that $a$ can, in turn, be transferred to $t$ using a lowest-fee path from $v$ to $t$. Note that, $a_t=a$. If there exists no path from $v$ to $t$, then $a_v$ is set to $\infty$. Given an arc $(v,v')$, the value $a_{v'}$ is related to the value $a_v$ by the following, which is derived from Equation~(\ref{eq:recurrence_relation}):
\begin{equation} \label{eq:recurrence_relation_dijkstra_1}
	a_{v'} = a_{v} - f_{br}((v,v'), a_{v'}).
\end{equation}

Initially, $a_{v}$ is unknown for all $v \in V$ apart from $t$. However, applying Dijkstra's algorithm to the graph $G$ to construct a lowest-fee path from $s$ to $t$ requires knowledge of $a_s$. To illustrate this, consider the LN displayed in Figure~\ref{fig:dijkstra_eg} and the case where we wish to make a payment of amount $a=10$ from $s$ and to $t$. The vertex $s$ has two adjacent arcs $(s, i)$ and $(s, j)$ which have equal base fees but different fee rates. The value $a_t=10$; however, the value $a_s$ is initially unknown and, therefore, it is initially not possible to compute the values $a_i$ and $a_j$ by applying Equation~(\ref{eq:recurrence_relation_dijkstra_1}) to the arcs $(s,i)$ and $(s,j)$ respectively. Using a calculation described later, it turns out that $a_i=13$ and $a_j=30$. The higher latter value is a consequence of the high base fee and fee rate of the arc $(j,t)$ which means that a large amount would need to be transferred to the vertex $j$.

\begin{figure}
\begin{center}
\subfigure[]{\includegraphics[height=3.5cm]{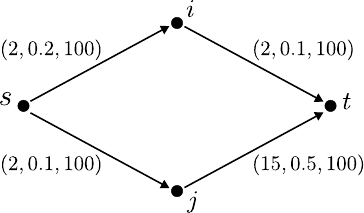}
\label{fig:dijkstra_eg}}
\hspace{1cm}
\subfigure[]{\includegraphics[height=3.5cm]{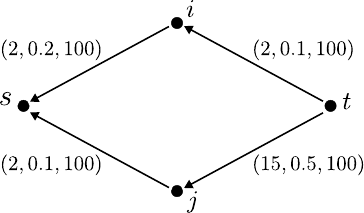}
\label{fig:dijkstra_eg_reverse}}
\caption{An example LN is displayed in (a) and the corresponding transpose, where arc directions are reversed, is displayed in (b).}
\end{center}
\end{figure}

This challenge can be overcome by constructing a lowest-fee path in the opposite direction from $t$ to $s$ instead of from $s$ to $t$. This approach is motivated by the fact that we initially only know the value $a_t$. Given an arc $(v,v')$, the value $a_{v}$ is related to the value $a_{v'}$ by the following:
\begin{equation} \label{eq:recurrence_relation_dijkstra_2}
	a_{v} = a_{v'} + f_{br}((v,v'), a_{v'}).
\end{equation}

By applying this equation recursively from the vertex $t$, we can now compute the value $a_v$ for all vertices $v$ where there exists a feasible path from $v$ to $t$. For example, consider again the LN displayed in Figure~\ref{fig:dijkstra_eg}. Using the fact that $a_t=10$, we can compute the values $a_i$ and $a_j$ by applying Equation~(\ref{eq:recurrence_relation_dijkstra_2}) to the arcs $(i,t)$ and $(j,t)$ respectively. These values in question are $10 + 2 + 0.1 \times 10 = 13$ and $10 + 15 + 0.5 \times 10 = 30$ respectively. Applying Equation~(\ref{eq:recurrence_relation_dijkstra_2}) recursively to the arcs $(s,i)$ and $(s,j)$, the value of $a_s$ is computed to be the minimum of $13 + 2 + 0.2 \times 13 = 17.6$ and $30 + 2 + 0.1 \times 30 = 35$, which equals $17.6$. Since the vertex $s$ has now been encountered, we have discovered a lowest-fee path from $s$ to $t$ which involves transferring a value of $17.6$ to vertex $s$, next transferring a value of $13$ to vertex $i$ and finally transferring a value of $10$ to vertex $t$.

The above path planning process corresponds to searching for a lowest-fee path from $t$ to $s$ in the transpose of the graph $G$, denoted by $G^T$. The transpose of a graph $G$ equals the graph $G$ in which arc directions have been reversed. Figure~\ref{fig:dijkstra_eg_reverse}, for example, shows the transpose of the graph from Figure~\ref{fig:dijkstra_eg}. In the following theorem, we prove that a lowest-fee path from $t$ to $s$ in $G^T$ corresponds to a lowest-fee path from $s$ to $t$ in $G$. Hence, the problem of determining a path with the latter property can be reformulated as the problem of determining a path with the former property. In turn, searching for a lowest-fee path from $t$ to $s$ in $G^T$ represents a correct approach.

\begin{theorem} \label{thm:reverse_search}
Let $G$ be a graph representation of the LN and let $G^T$ be its transpose. Let $(v_1, v_2)$, $(v_2, v_3), \dots, (v_{n-1}, v_n)$ be a lowest-fee path in $G^T$ from $t$ to $s$, where the amount transferred to $t$ is $a$. The transpose path $(v_n, v_{n-1}), (v_{n-1}, v_{n-2}), \dots, (v_2, v_1)$ in $G$ is a lowest-fee path from $s$ to $t$ where the amount transferred to $t$ is $a$.
\end{theorem}
\begin{proof}
Let $(v,v')$ be an arc in $G^T$ and $(v',v)$ be the corresponding transposed arc in $G$. Also, consider the path $(v,v')$ in $G^T$ and the path $(v',v)$ in $G$ where both paths contain a single arc. Consider the fee for the path $(v,v')$ in $G^T$ in the case where the amount transferred to $v$ is $a$ and fees are added to the amount transferred. This equals the fee for the path $(v',v)$ in $G$ where the amount transferred to $v$ is $a$ and fees are subtracted from the amount transferred. This result generalises from paths containing a single arc to paths containing multiple arcs and, in turn, proves the theorem.
\end{proof}

Algorithm~\ref{alg:dijkstra_alg} provides a pseudocode description of the above path planning process. It uses a map $c: V \rightarrow \mathbb{R}$ to model the total fee for the current lowest-fee path from $t$ to each vertex in $G^T$. Initially, all vertices are determined to have an infinite fee apart from $t$. A priority queue $Q$ is also used to order vertices by increasing values of $c(v)$. Initially, only $t$ is present in $Q$. As shown, the algorithm processes vertices in order of increasing fee by removing the minimum element from $Q$ at each step (lines 4 and 5). When the selected vertex $v$ is explored, the algorithm considers each neighbouring vertex $v'$ of $v$. The fee corresponding to the path from $t$ to $v'$ is the sum of the lowest-fee path from $t$ to $v$ plus the fee of the arc $(v,v')$. This is computed at line~9. If this fee is less than the fee of the current lowest-fee path from $t$ to $v'$ and the arc $(v,v')$ has a sufficient balance (line~10), then the maps and priority queue are updated on lines~11 and 12. If $v'$ is not yet present in $Q$, it is added to $Q$ with a fee of $c_{v'}$. The algorithm terminates either when $Q$ becomes empty (indicating that no path from $t$ to $s$ in $G^T$ exists) or when a lowest-fee path from $t$ to $s$ in $G^T$ is found (line~7). At termination, $p$ can be used to determine the lowest-fee path from $t$ to each vertex in $G^T$.

\begin{algorithm}
\caption{Variant of Dijkstra’s algorithm for LN}
\label{alg:dijkstra_alg}
\KwIn{A graph $G^T=(V,E)$ that equals the transpose of the LN graph $G$; a payment source $s \in V$; a payment destination $t \in V$; and a payment amount $a \in \mathbb{Z}^{>}$ to be transferred to $t$.}
\KwOut{A map $c: V \rightarrow \mathbb{R}$ that, for each vertex $v$, returns the fee for a lowest-fee path from $t$ to $v$ in $G^T$ where the amount transferred to $t$ is $a$; a map $p: V \rightarrow V$ that, for each vertex $v$, returns the previous vertex in a lowest-fee path from $t$ to $v$ in $G^T$.}
	For all $v\in V$, set $c(v)=\infty$\\
	Set $c(t)=0$ and insert the ordered pair $(c(t),t)$ into a priority queue $Q$\\
	\While{$|Q| > 0 $}{
		Let $(c(v), v)$ be the element in $Q$ with the minimum value for $c(v)$\\
		Remove the element $(c(v), v)$ from $Q$\\
		\If{$v=s$}{
			break
		}
		\ForAll{$v' \in \lbrace u : (v,u) \in E \rbrace$}{
			$c_{v'} = c(v) + f_{br}((v,v'), a)$ \\
			\If{$c_{v'} < c(v')$ $\land$ $c(v)+a \leq b(v,v')$}{
				Add the element $(c_{v'}, v')$ to $Q$ and, if present, remove element $(c(v'),v')$ from $Q$\\
				Set $c(v')=c_{v'}$ and set $p(v')=v$\\
			}
		}
	}
\end{algorithm}

Observe that Algorithm~\ref{alg:dijkstra_alg} is a variant of Dijkstra's original algorithm due to the presence of two features. Firstly, in this case, the computed cost of an arc $f_{br}((v,v'), a)$ (line~9) is dependent on the amount $a$ being transferred and is therefore not a constant. In contrast, Dijkstra's original algorithm assumes that all arc costs are non-negative constants. Secondly, in the above algorithm, an arc is only explored if it has a balance greater than or equal to the amount being transferred (line~10). On the other hand, Dijkstra's original algorithm contains no such condition. These differing features mean that existing correctness and computational complexity results for Dijkstra's original algorithm are insufficient here.

\section{Correctness and Computational Complexity Analysis}
\label{sec:formal_analysis}
In this section, we present proofs of the correctness and computational complexity of the LN path planning process in Algorithm~\ref{alg:dijkstra_alg}. 
We first introduce the definition of a \emph{consistent} fee map. This is a generalisation of the definition of a consistent travel time map introduced by Kaufman et al.~\cite{kaufman1993fastest}.
\begin{definition}
For a given LN graph $G=(V,E)$, a fee map $f: E \times \mathbb{Z}^{\geq} \rightarrow \mathbb{Z}^{\geq}$ is \emph{consistent} if for all $e \in E$, $a \in \mathbb{Z}$ and $a' \in \mathbb{Z}$ where $a \leq a'$, the following condition holds:
\begin{equation} \label{eq:consist}
	a + f(e,a) \leq a' + f(e,a').
\end{equation}
\end{definition}

Broadly speaking, this means that if a smaller amount is transferred along a given arc, the corresponding amount received will also be smaller. In the following theorem, we prove that the fee map $f_{br}$ defined in Equation~(\ref{eq:f_br}) and used within the LN is consistent.
\begin{theorem} \label{thm:fbr_consistent}
The fee map $f_{br}$ defined in Equation~\ref{eq:f_br} is a consistent fee map.
\end{theorem}
\begin{proof}
As noted by Kaufman et al.~\cite{kaufman1993fastest}, it follows from the definition of a consistent fee map that, a fee map is consistent if and only if its derivative with respect to $a$ is greater than or equal to $-1$. The derivative of $f_{br}(e,a) = f_b(e) + f_r(e) \times a$ with respect to $a$ is $f_r(e)$. Since this value is constrained to be greater than or equal to $0$, the result follows.
\end{proof}

Consider the fee map $f_{br}^\infty: E \times \mathbb{Z}^{\geq} \rightarrow \mathbb{Z}^{\geq}$ defined below. For a given arc $e$ and payment amount $a$, this map returns the value $f_{br}(e,a)$ if $e$ has a sufficient balance to transfer the amount $a$. Otherwise, this map returns $\infty$. This map is related to the concept of a barrier function in the field of optimisation \cite{nocedal2006nonlinear}. In the following lemma, we prove that the fee map $f_{br}^\infty$ is consistent. 
\begin{equation} \label{eq:f_br_infty}
f_{br}^\infty(e,a)=\begin{cases}
	f_{br}(e,a) & \text{if $a \leq b(e)$}.\\
	\infty & \text{otherwise}.
\end{cases}
\end{equation}

\begin{lemma} \label{thm:fbr_consistent_infty}
The fee map $f_{br}^\infty$ defined in Equation~\ref{eq:f_br_infty} is a consistent fee map.
\end{lemma}
\begin{proof}
To prove that the map $f_{br}^\infty$ is consistent we must prove that the statement $a + f_{br}^\infty(e,a) \leq a' + f_{br}^\infty(e,a')$ is true for all $e \in E$, $a \in \mathbb{Z}$ and $a' \in \mathbb{Z}$ where $a \leq a'$. The space of all variable combinations can be partitioned into two sets: $a \leq b(e)$ and $a > b(e)$. On the first set $a \leq b(e)$, $f_{br}^\infty(e,a)$ evaluates to $f_{br}(e,a)$. In this case, from Theorem~\ref{thm:fbr_consistent}, the statement $a + f_{br}^\infty(e,a) \leq a' + f_{br}^\infty(e,a')$ is true. On the second set $a > b(e)$, $f_{br}^\infty(e,a)$ evaluates to $\infty$. In this case, the statement $a + f_{br}^\infty(e,a) \leq a' + f_{br}^\infty(e,a')$ evaluates to $\infty \leq \infty$ and is also true.
\end{proof}

Toward proving the correctness and computational complexity of Algorithm~\ref{alg:dijkstra_alg}, we first show that the problem of determining a lowest-fee path in an LN with fee map $f_{br}^\infty$ is equivalent to the problem of determining a lowest time path in a time-dependent network (TDN). A TDN is a type of network where the time it takes to traverse a given connection is a function of the time this action is performed \cite{kaufman1993fastest, gendreau2015time}. For example, a street network can be modelled as a TDN, since the time it takes an agent to traverse a given street is a function of the time this action is performed. In general, a TDN is modelled as a directed graph, where arcs model transportation connections, plus a travel time map from arc and time pairs to travel times. Consider a TDN $G=(V,E)$ with travel time map $f_d: E \times \mathbb{Z}^{\geq} \rightarrow \mathbb{Z}^{\geq}$. Now consider the path $e_1, e_2$ in this TDN where $e_1 =(v_i, v_k)$, $e_2=(v_k, v_j)$ and the agent in question departs $v_i$ at time $a$. The agent arrives at $v_k$ at time $a + f_d(e_1, a)$. The agent in turn arrives at $v_j$ at time $a + f_d(e_1, a) + f_d(e_2, a + f_d(e_1, a))$. This computation is analogous to the LN payment computations described earlier in the previous section. We formalise this relationship in the following lemma. 

\begin{lemma} \label{thm:tdtn_ln_equ}
Consider a LN graph $G=(V,E)$ with fee map $f_{br}^\infty: E \times \mathbb{Z}^{\geq} \rightarrow \mathbb{Z}^{\geq}$ and a TDN graph $G'=(V',E')$ with time map $f_d: E' \times \mathbb{Z}^{\geq} \rightarrow \mathbb{Z}^{\geq}$. The problem of determining a lowest-fee path in $G$ from $s$ to $t$ (where the amount transferred to $t$ is $a$) is equivalent to the problem of determining a minimum time path in $G'$ from $t'$ to $s'$ (where the time the agent departs $t'$ is $a'$).
\end{lemma}
\begin{proof}
The problem of determining a lowest-fee path in $G$ from vertex $s$ to vertex $t$ where the amount transferred to $t$ is $a$ is equivalent to the problem of determining a lowest-fee path in $G^T$ from vertex $t$ to vertex $s$ where the amount transferred to $t$ is $a$ (see Theorem~\ref{thm:reverse_search}). By equating $G^T$ with $G'$, $s$ with $s'$, $t$ with $t'$, $a$ with $a'$ and $f_{br}^\infty$ with $f_d$, the path planning problems in the graphs $G'$ and $G^T$ are equivalent. In turn, the path planning problems in the graphs $G$ and $G'$ are equivalent.
\end{proof}

The equivalence established in Lemma~\ref{thm:tdtn_ln_equ} means that several results for path planning in TDNs generalise to path planning in an LN. This is a useful property, because it allows us to leverage existing research on the topic of TDNs that has been developed over several decades \cite{kaufman1993fastest, gendreau2015time}. Toward this goal, we define an algorithm denoted by Algorithm~$1'$. This algorithm modifies Algorithm~\ref{alg:dijkstra_alg} by replacing the fee map $f_{br}$ with $f_{br}^\infty$ and modifying the condition used to determine if a given arc is explored in the search for a lowest fee path. We subsequently prove the correctness and computational complexity of Algorithm~$1'$ before generalising these results to Algorithm~$1$.
\begin{definition}
Replace in Algorithm~\ref{alg:dijkstra_alg} the fee map $f_{br}$ with $f_{br}^\infty$ on line~9 and replace the statement $c_{v'} < c(v')$ $\land$ $c(v)+a \leq b(v,v')$ with $c_{v'} < c(v')$ on line~10. We denote this modified algorithm as Algorithm~$1'$. 
\end{definition}

\begin{lemma} \label{thm:uni_correctness}
Algorithm~$1'$ for computing a lowest-fee path in the LN graph $G=(V,E)$ with fee map $f_{br}^\infty$ is correct.
\end{lemma}
\begin{proof}
Algorithm~$1'$ is a variant of Dijkstra's algorithm. Kaufman et al.~\cite{kaufman1993fastest} prove that using this this type of algorithm variant to determine a minimum time path in a TDN with a consistent travel time map is correct (see Theorem~4 in the original paper by Kaufman et al.). Lemma~\ref{thm:fbr_consistent_infty} proves that the fee map $f_{br}^\infty$ is a consistent fee map. Lemma~\ref{thm:tdtn_ln_equ} proves that the problem of determining a lowest-fee path in an LN with fee map $f_{br}^\infty$ is equivalent to the problem of determining a minimum time path in a TDN with travel time map $f_d$. Combining the above three results, implies that Algorithm~$1'$ is correct.
\end{proof}

\begin{lemma} \label{thm:unidirectional_complexity}
The computational complexity of applying the path planning algorithm Algorithm~$1'$ to the LN graph $G=(V,E)$ with fee map $f_{br}^\infty$ is $O(|E| + |V| \log |V|)$.
\end{lemma}
\begin{proof}
It was proved by Kaufman et al. (see Theorem~4 in ~\cite{kaufman1993fastest}) that applying Dijkstra's algorithm to a TDN has equal computational complexity to applying Dijkstra's algorithm to a TDN where each arc has a constant and non-negative travel time. Hence, computational complexity results for applying Dijkstra's algorithm to graphs with constant and non-negative weights generalise to Algorithm~$1'$. The algorithm is applied to the transpose of the graph $G$ which contains an equal number of vertices and arcs. Using a Fibonacci heap data structure for the priority queue $Q$, an application of Dijkstra's algorithm to this graph has computational complexity $O(|E| + |V| \log |V|)$.
\end{proof}

The above results prove the correctness and computational complexity respectively of Algorithm~$1'$. If the following, we prove that these results generalise to Algorithm~\ref{alg:dijkstra_alg}. These proofs are based on the insight that Algorithm~\ref{alg:dijkstra_alg} prunes the search space of Algorithm~$1'$ to only consider feasible payment paths.

\begin{theorem}
Algorithm~\ref{alg:dijkstra_alg} for computing a lowest-fee path in the LN graph $G=(V,E)$ with fee map $f_{br}$ is correct.
\end{theorem}
\begin{proof}
In the following, we prove that Algorithms~\ref{alg:dijkstra_alg} and $1'$ compute equivalent results. The search space of Algorithm~\ref{alg:dijkstra_alg} is a pruned search space relative to Algorithm~$1'$. Specifically, the search space of Algorithm~\ref{alg:dijkstra_alg} is the space of feasible payment paths. The search space of Algorithm~$1'$ is the space of both feasible and infeasible payment paths. We assume that the range of $f_{br}$ is bounded and therefore all feasible payment paths have a finite fee. On the other hand, all infeasible payment paths will contain an arc where $f_{br}^\infty$ evaluates to $\infty$ and therefore have an infinite fee. If a feasible lowest fee path exists, since Algorithm~$1'$ is correct from Lemma~\ref{thm:uni_correctness}, Algorithm~$1'$ will return it. On the other hand, since it is feasible, this payment path will exist in the pruned search space and will also be returned by Algorithm~\ref{alg:dijkstra_alg}. If no feasible payment path exists, Algorithm~$1'$ will return an infeasible payment path with an infinite fee. On the other hand, Algorithm~\ref{alg:dijkstra_alg} will not return a payment path. If we equate returning an infeasible payment path with not returning a payment path, then Algorithms~\ref{alg:dijkstra_alg} and $1'$ compute equivalent results. Therefore, the correctness of Algorithm~$1'$ implies the correctness of Algorithm~\ref{alg:dijkstra_alg}.
\end{proof}

\begin{theorem}
The computational complexity of applying the path planning algorithm presented in Algorithm~\ref{alg:dijkstra_alg} to the LN graph $G=(V,E)$ with fee map $f_{br}$ is $O(|E| + |V| \log |V|)$.
\end{theorem}
\begin{proof}
Since Algorithm~\ref{alg:dijkstra_alg} searches a pruned search space relative to Algorithm~$1'$, it will have computational complexity less than or equal to Algorithm~$1'$. We assume that payment amounts and the range of $f_{br}$ are both bounded. Therefore, the amount that one may wish to forward along an arc is bounded. In a worst case where each arc has a balance greater than or equal to this bound, the search spaces of Algorithm~\ref{alg:dijkstra_alg} and Algorithm~$1'$ are equal. Hence Algorithm~\ref{alg:dijkstra_alg} and Algorithm~$1'$ have equal computational complexity. Given this, the proof follows from Lemma~\ref{thm:unidirectional_complexity}.
\end{proof}

The above proofs of correctness and computational complexity for Algorithm~\ref{alg:dijkstra_alg} assume that the fee map in question is consistent. If this were not the case, the algorithm would not be correct. To demonstrate this, consider the LN graph $G$ and corresponding transpose graph $G^T$ displayed in Figures~\ref{fig:consistent_map_eg_a} and \ref{fig:consistent_map_eg_b} respectively. Furthermore, let the fee map $f_{br}$ corresponding to the transpose graph $G^T$ have the properties $f_{br}((t,i), 100)=10$, $f_{br}((i,j), 110)=10$, $f_{br}((j,s), 120)=5$, $f_{br}((t,j), 100)=10$ and $f_{br}((j,s), 110)=20$. This fee map is not consistent because the following inequality is not satisfied $110 + f_{br}((j,s), 110)= 110+20=130 \allowbreak \leq \allowbreak 120 + f_{br}((j,s), 120) = 120 + 5 = 125$.

Consider the problem of making a payment from $s$ to $t$ in $G$ where the amount transferred to $t$ is 100. Applying Algorithm~\ref{alg:dijkstra_alg} to $G^T$ to determine a path from $t$ to $s$ returns the path $(t,j), (j,s)$ which has a fee of $f_{br}((t,j), 100) \allowbreak + f_{br}((j,s), 110) \allowbreak = \allowbreak 10 + 20 = \allowbreak 30$. However, the lowest fee path from $t$ to $s$ is in fact $(t,i), (i,j), (j,s)$, which has a fee of $f_{br}((t,i), 100) \allowbreak + f_{br}((i,j), 110) \allowbreak + f_{br}((j,s), 120) \allowbreak = 10 + 10 + 5 \allowbreak = 25$. This is because when the algorithm processes the arc $(t,j)$ it will insert the element $(110, j)$ into the queue $Q$ and consider the path $(t,j)$ as the lowest fee path from $t$ to $j$. When the algorithm later processes the arc $(i,j)$ it will observe that $120 \nleq 110$ and, consequently will not insert the element $(120, i)$ into the queue $Q$ nor consider the path $(t,i), (i,j)$ as the lowest fee path from $t$ to $j$. Hence, Algorithm~\ref{alg:dijkstra_alg} does not return the lowest fee path and is therefore incorrect in this case.

\begin{figure}
\begin{center}
\subfigure[]{\includegraphics[height=2.4cm]{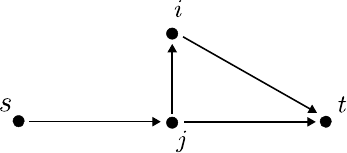}
\label{fig:consistent_map_eg_a}}
\hspace{1cm}
\subfigure[]{\includegraphics[height=2.4cm]{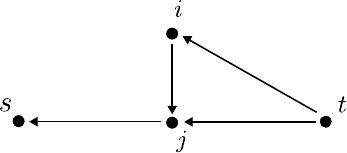}
\label{fig:consistent_map_eg_b}}
\caption{An LN graph $G$ and corresponding transpose graph $G^T$ displayed in (a) and (b) respectively. The fee map $f_{br}$ corresponding to $G^T$ has the properties $f_{br}((t,i), 100)=10$, $f_{br}((i,j), 110)=10$, $f_{br}((j,s), 120)=5$, $f_{br}((t,j), 100)=10$ and $f_{br}((j,s), 110)=20$.}
\end{center}
\end{figure}

If the fee map for a LN graph is an arbitrary map that is not consistent, the corresponding path planning can be NP-hard. We present a proof of this fact in the following theorem.
\begin{theorem}
Consider a LN graph $G=(V,E)$ with fee map $f: E \times \mathbb{Z}^{\geq} \rightarrow \mathbb{Z}^{\geq}$. If $f$ is an arbitrary map that is not consistent, the problem of determining a lowest-fee path in $G$ from $s$ to $t$, where the amount transferred to $t$ is $a$, can be NP-hard.
\end{theorem}
\begin{proof}
In Lemma~\ref{thm:tdtn_ln_equ} we proved an equivalence between the path planning problem in question and a specific type of path planning problem in a TDN. The latter problem is denoted \emph{forbidden waiting} by Orda et al.~\cite{orda1990shortest}, and Zeitz~\cite{zeitz2023np} proved that this problem is NP-hard. This implies the path planning problem in question is also NP-hard.
\end{proof}

\section{Bidirectional Variant of Dijkstra’s algorithm}
\label{sec:bi_path_find}
The path planning algorithm described in the previous section can be considered  \emph{unidirectional} in that it searches for a path in a single direction from the destination vertex to the source vertex. On the other hand, bidirectional path planning algorithms search for a path in two directions simultaneously. The first search takes place from the destination vertex towards the source vertex, the second search takes place from the source vertex towards the destination vertex. A path is then established when both searches intersect \cite{russell2021artificial}. If applicable to the LN, bidirectional path planning has the potential to reduce the time and space requirements relative to unidirectional path planning. This has been demonstrated empirically in many types of real-world networks including street networks \cite{bast2016route}.

To illustrate how bidirectional path planning can reduce the computational effort needed for path planning in the LN, consider the example LN displayed in Figure~\ref{fig:hub_spoke_ln} which has a hub-and-spoke topology. Although this is not an accurate representation of the real LN, the LN is known to exhibit topologies with hub-and-spoke structures \cite{martinazzi2020evolving}. In this context, vertices on the spokes likely correspond to customers and merchants, while vertices in the hub correspond to routing vertices. In this example, consider the case where we wish to transfer an amount $a$ from vertex $s$ to vertex $t$ and let us assume that all arcs have sufficient balance to transfer this amount plus any necessary fees. There is a single path $(s,r), (r, t)$ from $s$ to $t$ which contains a single intermediate vertex $r$. If we apply the unidirectional path planning algorithm described in the previous section, this algorithm will explore every vertex adjacent to $r$ before finding the lowest-fee path. On the other hand, if we could apply a bidirectional path planning algorithm, this algorithm would only need to explore the three vertices $s$, $r$ and $t$ to find the lowest-fee path. This is because both searches intersect when they explore the vertex $r$ and therefore do not need to explore any other vertices on the spokes.

\begin{figure}
\begin{center}
\includegraphics[height=6cm]{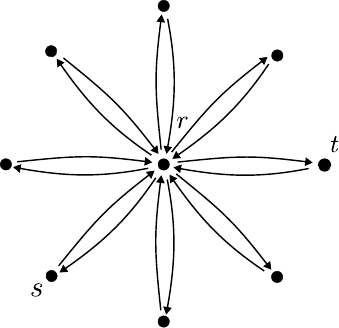}
\caption{An example LN with a hub-and-spoke topology is displayed.}
\label{fig:hub_spoke_ln}
\end{center}
\end{figure}

The asymptotic computational complexity of both the unidirectional and bidirectional path planning algorithms are both $O(|E| + |V| \log |V|)$. The former corresponds to applying Dijkstra’s algorithm once, while the latter corresponds to applying Dijkstra’s algorithm twice. However, the empirical benefits of the bidirectional algorithm can be explained by considering the branching factor $b$ of the searches and the number of arcs $d$ in the lowest-fee path. In graph searching algorithms such as Dijkstra's algorithm, the branching factor equals the average number of branching paths the search can explore from a given vertex \cite[Chapter~3]{russell2021artificial}. This approximately equals the average vertex degree. If a path planning algorithm is applied, the number of vertices explored will be approximately $O(b^d)$. This follows from the fact that the search will explore a search tree to a depth $d$ before the lowest-fee path is found and at each step, the search tree will expand by a factor of $b$. On the other hand, if a bidirectional path planning algorithm is applied, the number of vertices explored will be approximately $2 \times O(b^{d/2})$ which is less than $O(b^d)$ \cite[Chapter~3]{russell2021artificial}. Here, the multiple $2$ represents the fact that two searches are applied. The $d/2$ exponent represents the fact that each of these searches will explore a search tree to a depth of $d/2$ before the lowest-fee path is found. The above analysis implies that, if bidirectional path planning could be successfully applied to the LN, this would reduce the number of vertices explored. As discussed in the previous section, applying a path planning algorithm from the destination vertex $t$ to the source vertex $s$ is possible using Algorithm~\ref{alg:dijkstra_alg}. However, applying a path planning algorithm in the opposite direction is challenging.

To partially overcome this challenge, we can use the fact that fees charged by outgoing arcs adjacent to the vertex $s$ are paid to $s$. That is, any fees paid by the vertex $s$ will automatically be refunded and therefore can be considered to be zero. This fact can be modelled as a transformation of the LN such that $f_b(s,v)=0$ and $f_r(s,v)=0.0$ for all $v \in \lbrace u : (s,u) \in E \rbrace$. For example, applying this transformation to the LN in Figure~\ref{fig:dijkstra_eg} gives the LN displayed in Figure~\ref{fig:dijkstra_eg_transform}. Note that the transformation does not alter arc balances. As a consequence of the fact that the LN does not support negative fees, the value $0.0$ is a tight lower bound for the fees charged for transferring any amount. Therefore, when a vertex $v$ in the set $\lbrace u : (s,u) \in E \rbrace$ has been explored in a unidirectional search from $t$ and the arc $(s,v)$ has sufficient balance, a lowest cost path from $t$ to $s$ can be inferred by appending the arc $(s,v)$ to the lowest-fee path from $v$ to $t$. In this case, the search can be terminated before exploring the vertex $s$. For example, consider again the LN  in Figure~\ref{fig:dijkstra_eg}. When the vertex $i$ has been explored in a unidirectional search from $t$, a lowest cost path from $t$ to $s$ can be inferred by appending the arc $(s,i)$ to the lowest cost path from $i$ to $t$.

\begin{figure}
\begin{center}
\includegraphics[height=3.5cm]{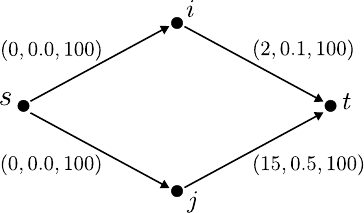}
\caption{An example LN is displayed.}
\label{fig:dijkstra_eg_transform}
\end{center}
\end{figure}

Using this insight, we can extend the termination condition used in the unidirectional path planning method shown in Algorithm~\ref{alg:dijkstra_alg}. Specifically, instead of just terminating when the next vertex $v$ to be explored is the source $s$ (line~6), we can also terminate when $v$ belongs to the set $\lbrace u : (u,s) \in E \rbrace$ and the arc in question has a sufficient balance (line~8). Algorithm~\ref{alg:dijkstra_alg_bi} provides a full description of this new algorithm. This is the same as  Algorithm~\ref{alg:dijkstra_alg} except for the additional termination condition at line~8. We now prove the correctness of this algorithm.

\begin{theorem} \label{thm:bi_correctness}
Algorithm~\ref{alg:dijkstra_alg_bi} for computing a lowest-fee path in the LN is correct.
\end{theorem}
\begin{proof}
It follows from Theorem~\ref{thm:uni_correctness} that the algorithm computes a lowest-fee path from the vertex $t$ to the vertex $s$ or a vertex in the set $v \in \lbrace u : (u,s) \in E \rbrace$. When either of these events occurs, the algorithm, in turn, computes a lowest-fee path from $t$ to $s$ in $G^T$. Hence, Algorithm~\ref{alg:dijkstra_alg_bi} is correct.
\end{proof}

\begin{algorithm}
\caption{Bidirectional Variant of Dijkstra’s algorithm}
\label{alg:dijkstra_alg_bi}
\KwIn{A graph $G^T=(V,E)$ that equals the transpose of the LN graph $G$; a payment source $s \in V$; a payment destination $t \in V$; and a payment amount $a \in \mathbb{Z}^{>}$ to be transferred to $t$.}
\KwOut{A map $c: V \rightarrow \mathbb{R}$ that, for each vertex $v$, returns the fee for a lowest-fee path from $t$ to $v$ in $G^T$ where the amount transferred to $t$ is $a$; a map $p: V \rightarrow V$ that, for each vertex $v$, returns the previous vertex in a lowest-fee path from $t$ to $v$ in $G^T$.}
For all $v\in V$, set $c(v)=\infty$\\
Set $c(t)=0$ and insert the ordered pair $(c(t),t)$ into a priority queue $Q$\\
	\While{$|Q| > 0 $}{
		Let $(c(v), v)$ be the element in $Q$ with the minimum value for $c(v)$\\
		Remove the element $(c(v), v)$ from $Q$\\
		\If{$v=s$}{
			break\\
		}
		\If{$v\in\{u:(u,s)\in E\}$ $\land$ $c(v)+a\leq b(v,s)$}{
			Set $c(s)=c(v)$, set $p(s)=v$, and break\\
		}
		\ForAll{$v' \in \lbrace u : (v,u) \in E \rbrace$}{
			$c_{v'} = c(v) + f_{br}((v,v'), a)$ \\
			\If{$c_{v'} < c(v')$ $\land$ $c(v)+a \leq b(v,v')$}{
				Add the element $(c_{v'}, v')$ to $Q$ and, if present, remove element $(c(v'),v')$ from $Q$\\
				Set $c(v')=c_{v'}$ and set $p(v')=v$\\
			}
		}
	}
\end{algorithm}

Algorithm~\ref{alg:dijkstra_alg_bi} is a variant on bidirectional path planning in which a full search is initiated from $t$ and only a partial search, consisting of a single step, is initiated from $s$. Hence, we refer to this path planning algorithm as a \emph{partial bidirectional path planning} algorithm. As discussed earlier, a bidirectional path planning algorithm explores fewer vertices than a unidirectional path planning algorithm. Since the partial bidirectional path planning terminates the search one step earlier, it will explore approximately $O(b^{d-1})$ where $b$ is the branching factor and $d$ is the number of arcs in the lowest-fee path. This is less than that of the unidirectional path planning algorithm equals $O(b^{d})$ but greater than that of the bidirectional path planning algorithm, which equals $2 \, O(b^{d/2})$. The number of vertices explored by the partial bidirectional algorithm is not significantly less than that for unidirectional path planning. However, as discussed in the results section below, we find in practice that partial bidirectional path planning exhibits better performance. This can be attributed to the hub and spoke topology of the LN described earlier.

\section{Experimental Results}
\label{sec:results}
In this section, we present an empirical evaluation of the unidirectional and partial bidirectional path planning algorithms described in Sections~\ref{sec:uni_path_find} and \ref{sec:bi_path_find} respectively. As discussed in the Related Works section, existing path planning algorithms can broadly be categorised as centralised algorithms and decentralised algorithms. The unidirectional algorithm is the most commonly used centralised algorithm and therefore is a suitable benchmark to compare the partial bidirectional algorithm against. On the other hand, decentralised algorithms are not commonly used due to the lack of payment privacy they provide. In turn, implementations of these methods are difficult to obtain. Therefore, we do not consider any decentralised algorithms in our comparison.

In our trials, we used a real snapshot of the LN captured on 22 February 2024 using the Lightning Network Daemon (lnd) implementation of an LN node.\footnote{https://docs.lightning.engineering/} On this same date, one Satoshi (the unit of currency used by the LN) had a value of approximately 0.0004 GBP. The graph representation of the above snapshot contains 9,420 vertices and 32,202 arcs; however, many of these arcs had no fee information making them unusable for routing and were therefore removed. Furthermore, many vertices were isolated with no adjacent arcs and these were also removed. Following these steps, the resulting graph contained 2,453 vertices and 26,000 arcs. The mean and median LN snapshot channel capacities were 10,490,032 Satoshis and 4,000,000 Satoshis respectively. The 10th and 90th percentiles of LN snapshot channel capacities were 500,000 Satoshis and 17,137,075 Satoshis respectively. To help maintain privacy and reduce communication bandwidth, individual arc balances are not shared. We therefore assumed that all channels were balanced. That is, to each of the two arcs corresponding to a given channel, we assigned a balance equal to half of the corresponding channel capacity. Currently, there is no historical transaction data for the LN and such data is difficult to infer by design. Therefore, to test the two path planning algorithms, we constructed three sets of simulated payments on the LN snapshot described above. The first set of payments contained 10,000 elements and was constructed using the following rejection sampling-based approach. We first sampled without replacement the payment source and destination vertices from the set of graph vertices $V$. This ensures that the pair of vertices is distinct. Next, we sampled the payment amount uniformly at random from the set $\{1,2,\ldots, 1000000\}$. A recent report by the cryptocurrency company River indicates that most LN payment amounts lie in this interval \cite{river2023}. Given the above payment source, destination, and amount, we then determined if the payment in question was feasible; that is, whether there existed a path from the source to the destination that could transfer the amount. If the payment was feasible, it was added to the set of payments. Otherwise, it was discarded. This process was repeated until the set of payments contained 10,000 elements.

For each simulated payment, we computed the number of vertices explored by the unidirectional and partial bidirectional algorithms by counting the number of corresponding for-loop iterations. (The loops in question start at lines~8 and 10 of Algorithms \ref{alg:dijkstra_alg} and \ref{alg:dijkstra_alg_bi} respectively). For the unidirectional algorithm, the mean and standard deviation of the number of vertices explored was $7,452 \pm 3,190$. For the partial bidirectional algorithm, the corresponding figures were $4,069 \pm  3,594$, equating to a 45\% reduction in the mean. Next, for each payment, we computed the percentage reduction in vertices explored by the partial bidirectional algorithm. That is, the number explored by unidirectional minus the number explored bidirectional divided by the number explored by unidirectional. This percentage represents the reduction in computation provided by the proposed bidirectional algorithm for the path planning task in question. Figure~\ref{fig:uniform_hist} displays a histogram of these percentage values. The mean and standard deviation of these percentages are 47\% and 35\% respectively.

\begin{figure}
\begin{center}
\subfigure[]{\includegraphics[height=4cm]{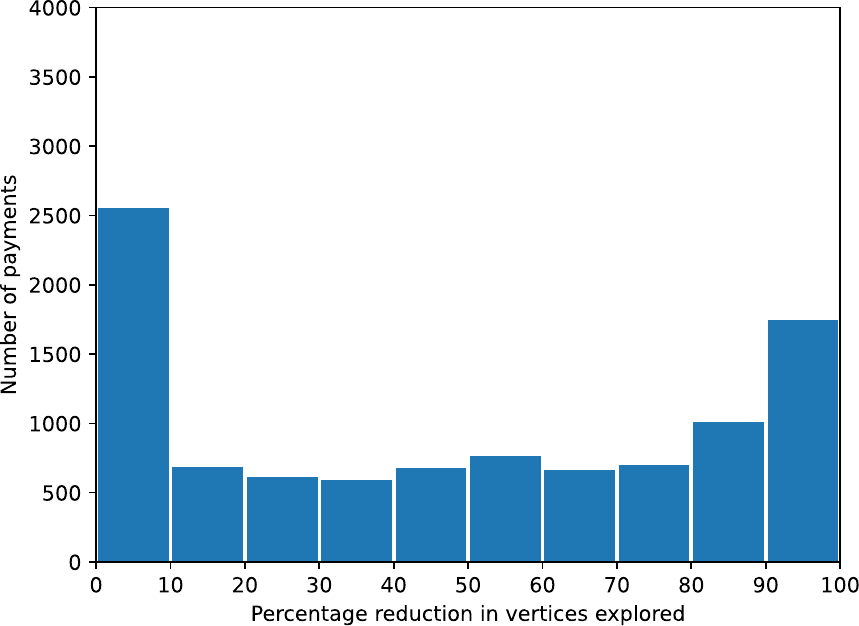}
\label{fig:uniform_hist}}
\hspace{.5cm}
\subfigure[]{\includegraphics[height=4cm]{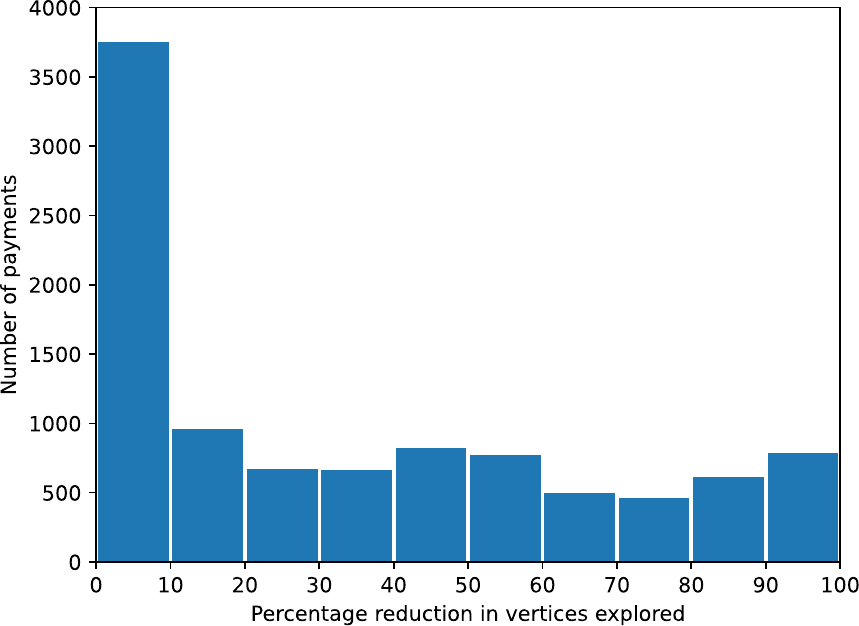}
\label{fig:edge_hist}}
\caption{For two sets of simulated payments, the corresponding histograms of the percentage reduction in the number of vertices explored by the partial bidirectional algorithm relative to the unidirectional algorithm are displayed. The vertical axis is the number of payments, while the horizontal axis gives the percentage reduction.}
\end{center}
\end{figure}

Our second set of simulated payments also contained 10,000 elements. This set was constructed using an approach similar to the previous though, in this case, we sampled the source and destination vertices from the set of all graph vertices with out-degrees of less than four. The motivation for considering these low-degree vertices is that they tend to correspond to merchants and customers instead of routing peers. Similarly to before, we then computed the percentage reduction in the number of vertices explored by the partial bidirectional algorithm. In this case, the mean number of vertices explored by the unidirectional algorithm was $8,102 \pm 2,882$ whereas, for the partial bidirectional algorithm, these figures were $5,491 \pm 3,574$, giving a 32\% reduction in the mean. For each payment, we then computed the percentage reduction in vertices explored by the partial bidirectional algorithm. Figure~\ref{fig:edge_hist} displays a histogram of these values. The mean and standard deviation of these percentage values are 33\% and 32\% respectively.

Our final set of simulated payments contained 100,000 elements. This set was constructed using the same approach as the first set. We then used this set of payments to measure the wall-clock running time of both path planning algorithms. Both algorithms were implemented as single-thread computer programs using the Rust programming language and executed on a PC with an Inter i9-14900K CPU and 32GB of RAM. We computed the set of payments using each algorithm and measured the corresponding wall-clock running times. We then repeated this step ten times and computed the corresponding mean wall-clock running times. The mean wall-clock running times for the unidirectional and partial bidirectional path planning algorithms were 66 and 48 seconds respectively. This equates to a mean wall-clock running time for an individual payment of 0.00066 and 0.00048 seconds respectively. The latter algorithm therefore provides a 27\% reduction in mean wall-clock running time.

In summary, for all three sets of simulated payments considered, the partial bidirectional path planning algorithm explores significantly fewer vertices than the unidirectional path planning algorithm and, in turn, provides improved computational performance.

\section{Conclusions}
\label{sec:conclusions}
In this article, we have considered the problem of path planning in PCNs. We have provided correctness and computational complexity results for a variant of Dijkstra's algorithm, which is the most cited algorithm for solving this problem. These results show that if the PCN fee function has a consistency property, the path planning algorithm is correct and has polynomial computational complexity. This finding therefore has the potential to inform the development of new or existing PCN fee functions. In this article, we have also proposed a slight modification to the above algorithm based on the concept of bidirectional search. We demonstrate empirically that this proposed modification provides improved computational performance.

These findings notwithstanding, solving the path planning problem is only one of many challenges that must be overcome before PCNs can fully deliver on their potential to solve the scaling challenges of cryptocurrencies \cite{dotan2021survey}. Other challenges include mitigating security attacks such as denial-of-service \cite{shikhelman2022unjamming} and griefing \cite{mazumdar2022strategic} attacks, ensuring the privacy of users \cite{kappos2021empirical}, ensuring network topologies that allow efficient payments \cite{guasoni2023lightning}, and allowing users with edge devices such as mobile phones to make and receive payments. In future work, we hope to help address some of these other challenges.

\section*{Acknowledgements}
This work was supported by the Security, Crime and Intelligence Innovation Institute at Cardiff University through their Kickstarter Funding scheme.

\bibliographystyle{unsrt}
\bibliography{pcn_correctness_complexity_analysis}

\appendix
\section{Recurrence Relation}
In the following theorem, we solve the recurrence relation defined in Equation~(\ref{eq:recurrence_relation}) to give a closed-form expression for each $a_i$ value.
\begin{theorem}
The solution of the recurrence relation in Equation~(\ref{eq:recurrence_relation}) with the condition $a_{n+1}=a$ is
\begin{equation} \label{eq:recurrence_relation_solution}
a_{i} = \left( \prod_{k=i+1}^{n+1} \left( 1+f_r(e_k) \right)  \right) \left( a_{n+1} + \sum_{m=i+1}^{n+1} \frac{f_b(e_m)}{\prod_{k=m}^{n+1}1+f_r(e_k)} \right).
\end{equation}
\end{theorem}

\begin{proof}
Let $e'_1, e'_2, \dots, e'_n$ be the sequence of arcs where $e'_i=e_{n+1-i}$, and let $a'_1, a'_2, \dots, a'_{n+1}$ be the sequence of values where $a'_i=a_{n+2-i}$. That is, we reverse the order of both sequences. By performing this change of variables, Equation~(\ref{eq:recurrence_relation}) can be rewritten as the following recurrence relation with the initial condition $a'_{1}=a$:
\begin{equation} \label{eq:recurrence_relation_change}
a'_{i+1} = a'_i + f_b(e'_i) + f_r(e'_i) a'_i.
\end{equation}

The recurrence relation in Equation~(\ref{eq:recurrence_relation_change}) can be refactored as
\begin{subequations}
	\begin{align}
	a'_{i+1} & = a'_i + f_b(e'_i) + f_r(e'_i) a'_i \\
	& = (1 + f_r(e'_i))a'_i + f_b(e'_i). \label{eq:recurrence_relation_proof_b}
\end{align}
\end{subequations}
The recurrence relation in Equation~(\ref{eq:recurrence_relation_proof_b}) has the properties of being linear, first-order, non-homogeneous and having variable coefficients \cite{charalambides2002enumerative}. First-order refers to the fact that the term $a'_{i+1}$ is a function of only the previous term $a'_i$ in the sequence. Non-homogeneous means that the term $f_b(e_i')$ is not a multiple of $a'_{i-1}$ or $a'_i$. Finally, having variable coefficients refers to the fact that the coefficients $f_b(e_i')$ and $(1 + f_r(e_i'))$ are a function of $i$. A direct application of Theorem~7.1 in \cite{charalambides2002enumerative} solves this recurrence relation, where the solution is defined as
\begin{equation}
a'_{i} =  \left( \prod_{k=1}^{i-1} \left( 1+f_r(e'_k) \right)  \right) \left( a'_{1} + \sum_{m=1}^{i-1} \frac{f_b(e'_m)}{\prod_{k=1}^{m}1+f_r(e'_k)} \right).
\end{equation}

Performing an inverse of the change of variables previously performed gives the following solution with the initial condition $a_{n+1}=a$:
\begin{equation}
a_{n+2-i} = \left( \prod_{k=n+1}^{n+1-i} \left( 1+f_r(e_k) \right)  \right) \left( a_{n+1} + \sum_{m=n+1}^{n+1-i} \frac{f_b(e_m)}{\prod_{k=n+1}^{m}1+f_r(e_k)} \right)
\end{equation}
implying that
\begin{align}
a_{i} &= \left( \prod_{k=n+1}^{i+1} \left( 1+f_r(e_k) \right)  \right) \left( a_{n+1} + \sum_{m=n+1}^{i+1} \frac{f_b(e_m)}{\prod_{k=n+1}^{m}1+f_r(e_k)} \right) \\
&= \left( \prod_{k=i+1}^{n+1} \left( 1+f_r(e_k) \right)  \right) \left( a_{n+1} + \sum_{m=i+1}^{n+1} \frac{f_b(e_m)}{\prod_{k=m}^{n+1}1+f_r(e_k)} \right)
\end{align}
as required.
\end{proof}

The following Python code listing implements the above method for directly solving the recurrence relation in question. By comparing solutions computed using the recursive and direct methods defined in Equations~(\ref{eq:recurrence_relation}) and (\ref{eq:recurrence_relation_solution}) respectively, this code can be used to empirically verify the correctness of the above theorem.

\begin{lstlisting}[language=Python, caption=Solving the recurrence relation recursively and directly.]
import numpy

a = 102 #Amount a
b = [10, 5, 3.4, 11, 7] # path arc base fees
r = [0.1, 0.211, 0.15, 0.12, 0.11] # path arc fee rates
n = len(b)-1

#Compute the sequence of values recursively using Equation (2)
ai = a
print("a[", n+1, "]: ", ai)
for i in range(len(b)-1, -1, -1):
    ai += b[i] + ai*r[i]
    print("a[", i, "]: ", ai)
print()

#Compute the sequence of values directly using Equation (6)
for i in range(n+1, -1, -1):
    ai = numpy.prod([(1+r[k]) for k in range(i,n+1)]) * (a + numpy.sum([(b[m]/numpy.prod([(1+r[k]) for k in range(m,n+1)])) for m in range(i,n+1)]))
    print("a[", i, "]: ", ai)
\end{lstlisting}

\end{document}